\documentclass[envcountsame]{llncs}
    \usepackage{amsmath}
    \usepackage{amsxtra}
    \usepackage{amstext}
    \usepackage{amssymb}
    \usepackage{latexsym}
    \usepackage{dsfont} % for \mathds{N}
\newcommand{\tr}[2][1]{Tr_{#1}^{#2}}

\newcommand{\GF}[1]{{\mathbb F}_{#1}}

\begin{document}
\title{Complete solution over $\GF{p^n}$ of the equation $X^{p^k+1}+X+a=0$}
\author{Kwang Ho Kim\inst{1,2}\and Jong Hyok Choe\inst{1} \and Sihem Mesnager\inst{3}}
\institute{ Institute of Mathematics, State Academy of Sciences,
Pyongyang, Democratic People's Republic of Korea\\
\email{khk.cryptech@gmail.com} \and PGItech Corp., Pyongyang, Democratic People's Republic of Korea\\ \and Department of Mathematics, University of Paris VIII, F-93526 Saint-Denis, University
Sorbonne Paris Cit\'e, LAGA, UMR 7539, CNRS, 93430 Villetaneuse and T\'el\'ecom Paris, 91120 Palaiseau, France.\\
\email{smesnager@univ-paris8.fr}\\} \maketitle
\date{today}

\begin{abstract}
The problem of solving  explicitly the equation $P_a(X):=X^{q+1}+X+a=0$ over the finite
field $\GF{Q}$, where $Q=p^n$, $q=p^k$ and $p$ is a prime, arises in
many different contexts including finite geometry, the inverse
Galois problem \cite{ACZ2000}, the construction of difference sets
with Singer parameters \cite{DD2004}, determining cross-correlation
between $m$-sequences \cite{DOBBERTIN2006} and to construct error
correcting codes \cite{Bracken2009},  cryptographic APN functions
\cite{BTT2014,Budaghyan-Carlet_2006}, designs \cite{Tang_2019}, as
well as to speed up the index calculus method for computing discrete
logarithms on finite fields \cite{GGGZ2013,GGGZ2013+} and on
algebraic curves \cite{M2014}.

 Subsequently, in
\cite{Bluher2004,HK2008,HK2010,BTT2014,Bluher2016,KM2019,CMPZ2019,MS2019,KCM19},
the $\GF{Q}$-zeros  of $P_a(X)$ have been studied. In
\cite{Bluher2004}, it was shown that the possible values of the
number of
 the zeros that $P_a(X)$ has in $\GF{Q}$ is $0$, $1$, $2$ or $p^{\gcd(n, k)}+1$.
 Some criteria for the number of the $\GF{Q}$-zeros of $P_a(x)$  were
 found in \cite{HK2008,HK2010,BTT2014,KM2019,MS2019}.
However, while the ultimate goal is to explicit all the
$\GF{Q}$-zeros,
 even in the case $p=2$, it was solved only under the condition $\gcd(n, k)=1$ \cite{KM2019}.\\

In this article, we discuss this equation without any restriction on $p$ and
$\gcd(n,k)$. In \cite{KCM19}, for the cases of one or two
$\GF{Q}$-zeros, explicit expressions for these rational zeros in
terms of $a$ were provided, but for the case of $p^{\gcd(n, k)}+1$
$\GF{Q}-$ zeros  it was remained open to explicitly compute the
zeros. This paper solves the remained problem, thus now the equation
$X^{p^k+1}+X+a=0$ over $\GF{p^n}$ is completely solved for any prime
$p$, any integers $n$ and $k$.\\



\end{abstract}

\section{Introduction}
Let $n$ and $k$ be any positive integers with $\gcd(n,k)=d$. Let
$Q=p^n$ and $q=p^k$ where $p$ is a prime. We consider the polynomial
\[P_a(X):=X^{q+1}+X+a, a\in \GF{Q}^{*}.\] Notice the more general polynomial forms
$X^{q+1}+rX^{q}+sX+t$ with $s\neq r^q$ and $t\neq rs$ can be
transformed into this form by the substitution
 $X=(s-r^q)^{\frac{1}{q}}X_1-r$. It is clear that $P_a(X)$ have no multiple roots.

 These polynomials have arisen in several different contexts including finite geometry, the inverse Galois
problem \cite{ACZ2000}, the construction of difference sets with
Singer parameters \cite{DD2004}, determining cross-correlation
between $m$-sequences \cite{DOBBERTIN2006} and to construct error
correcting codes \cite{Bracken2009}, APN functions
\cite{BTT2014,Budaghyan-Carlet_2006}, designs \cite{Tang_2019}.
These polynomials are also exploited to speed up (the relation
generation phase in) the index calculus method for computation of
discrete logarithms on finite fields \cite{GGGZ2013,GGGZ2013+} and
on algebraic curves \cite{M2014}.

Let $N_a$ denote the number of zeros in $\GF{Q}$ of polynomial
$P_a(X)$ and $M_i$ denote the number of $a\in \GF{Q}^{*}$ such that
$P_a(X)$ has exactly $i$ zeros in $\GF{Q}$. In 2004, Bluher
\cite{Bluher2004} proved that $N_a$ takes either of 0, 1, 2 and
$p^d+1$ where $d=\gcd(k, n)$ and computed $M_i$ for every $i$. She
also stated some criteria for the number of the $\GF{Q}$-zeros of
$P_a(X)$.

The ultimate goal in this direction of research is to identify all
the $\GF{Q}$-zeros of $P_a(X)$. Subsequently, there were much
efforts for this goal, specifically for a particular instance of the
problem over binary fields i.e. $p=2$. In 2008 and 2010, Helleseth
and Kholosha \cite{HK2008,HK2010} found new criteria for the number
of $\GF{2^n}$-zeros of $P_a(X)$. In the cases when there is a unique
zero or exactly two zeros and $d$ is odd, they provided explicit
expressions of these zeros as polynomials of $a$ \cite{HK2010}. In
2014, Bracken, Tan, and Tan \cite{BTT2014} presented a criterion for
$N_a=0$ in $\GF{2^n}$ when $d=1$ and $n$ is even. In 2019, Kim and
Mesnager \cite{KM2019} completely solved this equation
$X^{2^k+1}+X+a=0$ over $\GF{2^n}$ when $d=1$. They showed that the
problem of finding zeros in $\GF{2^n}$ of $P_a(X)$, in fact, can be
divided into two problems with odd $k$: to find the unique preimage
of an element in $\GF{2^n}$ under an M\"{u}ller-Cohen-Matthews
polynomial and to find preimages of an element in $\GF{2^n}$ under a
Dickson polynomial. By completely solving these two independent
problems, they explicitly calculated all possible zeros in
$\GF{2^n}$ of $P_a(X)$, with new criteria for which $N_a$ is equal
to $0$, $1$ or $p^d+1$ as a by-product.

Very recently, new criteria for which $P_a(X)$ has $0$, $1$, $2$ or
$p^d+1$ roots were stated by \cite{KCM19,MS2019} for any
characteristic. In \cite{KCM19}, for the cases of one or two
$\GF{Q}$-zeros, explicit expressions for these rational zeros in
terms of $a$ are provides. For the case of $p^{\gcd(n, k)}+1$
rational zeros, \cite{KCM19} provides a parametrization of such
$a$'s and expresses the $p^{\gcd(n, k)}+1$ rational zeros by using
that parametrization, but it was remained open to explicitly
represent the zeros.

Following \cite{KCM19}, this paper discuss the equation
$X^{p^k+1}+X+a=0, a\in \GF{p^n}$, without any restriction on $p$ and
$\gcd(n,k)$. After introducing some prerequisites from \cite{KCM19}
(Sec. \ref{prereq}), we solve the open problem remained in
\cite{KCM19} to explicitly represent the $\GF{Q}-$zeros for the case
of $p^{\gcd(n, k)}+1$ rational zeros (Sec. \ref{complete}). After
all, it is concluded that the equation $X^{p^k+1}+X+a=0$ over
$\GF{p^n}$ is completely solved for any prime $p$, any integers $n$
and $k$.

\section{Prerequisites}\label{prereq}
Throughout this paper, we maintain the following notations.\\
\textbullet\quad $p$ is any prime. \\
\textbullet\quad  $n$ and $k$ are any positive integers. \\
\textbullet\quad $d=\gcd(n,k)$. \\
\textbullet\quad  $m:=n/d$.\\
\textbullet\quad $q=p^k$. \\
\textbullet\quad $Q=p^n$. \\
\textbullet\quad $a$ is any element of the finite field
$\GF{Q}^*$.\\

 Given positive integers $L$ and $l$, define a polynomial
\[T^{Ll}_{L}(X):=X+X^{p^L}+\cdots+X^{p^{L(l-2)}}+X^{p^{L(l-1)}}.\]
Usually we will abbreviate $T^{l}_{1}(\cdot)$ as $T_l(\cdot)$. For
$x\in \GF{p^l}$, $T_l(x)$ is the absolute trace $\tr{l}(x)$ of $x$.

In \cite{KCM19}, the sequence of polynomials $\{A_r(X)\}$ in
$\GF{p}[X]$ is defined as follows:
\begin{equation}\label{eq.defA}
\begin{aligned}
&A_1(X)=1, A_2(X)=-1,\\
&A_{r+2}(X)=-A_{r+1}(X)^{q}-X^{q}A_{r}(X)^{q^2} \text{ for } r\geq
1.
\end{aligned}
\end{equation}
The following lemma gives another identity which can be used as an
alternative definition of $\{A_r(X)\}$ and an interesting property
of this polynomial sequence which will be importantly applied
afterwards.
\begin{lemma}[\cite{KCM19}]\label{Lem.newdef}
 For any $r\geq 1$, the following are true.
\begin{enumerate}
\item \begin{equation}\label{eq.defAA}
A_{r+2}(X)=-A_{r+1}(X)-X^{q^r}A_{r}(X).
\end{equation}
\item \begin{equation}\label{eq.Norm}
A_{r+1}(X)^{q+1}-A_r(X)^qA_{r+2}(X)=X^{\frac{q(q^r-1)}{q-1}}.
\end{equation}
\end{enumerate}
\end{lemma}

The zero set of $A_r(X)$ can be completely determined for all $r$:
\begin{proposition}[\cite{KCM19}]\label{Prop.A_r(x)=0}
For any $r\geq 3$,
\[
\{x\in \overline{\GF{p}} \mid
A_r(x)=0\}=\left\{\frac{(u-u^q)^{q^2+1}}{(u-u^{q^2})^{q+1}}, \ \
u\in \GF{q^r}\setminus\GF{q^2}\right\}.
\]
%and $|Z_i| = \frac{q^{i-1}-1}{q^2-1}$ if $i$ is odd and $|Z_i| =
%\frac{q^{i-1}-q}{q^2-1}$ if $i$ is even.
\end{proposition}

Further, define polynomials
\[\begin{aligned}
&F(X):=A_m(X),\\
&G(X):=-A_{m+1}(X)-XA_{m-1}^q(X).
\end{aligned}\]

It can be shown that if $F(a)\neq0$ then the $\GF{Q}$-zeros of
$P_a(X)$ satisfy a quadratic equation and therefore necessarily
$N_a\leq 2$.
\begin{lemma}[\cite{KCM19}]\label{lem.quadEq}
Let $a\in \GF{Q}^{*}$. If $P_a(x)=0$ for $x\in \GF{Q}$, then
\begin{equation}\label{eq.quadEq}
F(a)x^2+G(a)x+aF^q(a)=0.
\end{equation}
\end{lemma}

By exploiting these definitions and facts, the following results
have been got.

\subsection{$N_a\leq 2$: Odd $p$}

\begin{theorem}[\cite{KCM19}]\label{Theo.oddp_N<3}
Let $p$ be odd. Let $a\in \GF{Q}$ and $E=G(a)^2-4a{F(a)}^{q+1}$.
\begin{enumerate}
\item $N_a=0$ if and only if $E$ is not a quadratic residue in $\GF{p^d}$ (i.e.
$E^{\frac{p^d-1}{2}}\neq 0,1).$
\item $N_a=1$ if and only if $F(a)\neq 0$ and $E=0.$ In this case, the
unique zero in $\GF{Q}$ of $P_a(X)$ is $-\frac{G(a)}{2F(a)}.$
\item $N_a=2$ if and only if  $E$ is a non-zero quadratic residue in $\GF{p^d}$ (i.e.
$E^{\frac{p^d-1}{2}}=1$). In this case, the two zeros in $\GF{Q}$ of
$P_a(X)$ are $x_{1,2}=\frac{\pm E^{\frac{1}{2}}-G(a)}{2F(a)}$, where
$E^{\frac{1}{2}}$ represents a quadratic root in $\GF{p^d}$ of $E$.
\end{enumerate}
\end{theorem}

\subsection{$N_a\leq 2$: $p=2$}

When $p=2$, in \cite{KCM19} it is proved that $ G(x)\in \GF{q}
\text{ for any } x\in \GF{q^m}$ and using it
\begin{theorem}[\cite{KCM19}]\label{Theo.p=2.N<3}
Let $p=2$ and $a\in \GF{Q}$. Let
$H=\tr{d}\left(\frac{{Nr}^n_d(a)}{G^2(a)}\right)$ and
$E=\frac{aF(a)^{q+1}}{G^2(a)}$.
\begin{enumerate}
\item $N_a=0$  if and only if  $G(a)\neq 0$ and $H\neq 0$.
\item $N_a=1$ if and only if $F(a)\neq 0$ and $G(a)=0$. In this case, $(aF(a)^{q-1})^{\frac{1}{2}}$ is the unique
zero in $\GF{Q}$ of $P_a(X)$.
\item $N_a=2$  if and only if  $G(a)\neq 0$ and $H=0$. In this
case the two zeros in $\GF{Q}$ are $x_1=\frac{G(a)}{F(a)}\cdot
T_n\left(\frac{E}{\zeta+1}\right)$ and $x_2=x_1+\frac{G(a)}{F(a)}$,
where $\zeta\in \mu_{Q+1}\setminus \{1\}$.
\end{enumerate}
\end{theorem}

\subsection{$N_a=p^d+1$: Auxiliary results}

 \begin{lemma}[\cite{KCM19}]\label{Lem.p^d+1cond} Let $a\in \GF{Q}^*$. The following are equivalent.
 \begin{enumerate}
\item\label{1} $N_a=p^d+1$ i.e. $P_a(X)$ has exactly $p^d+1$ zeros in $\GF{Q}$.
\item\label{2} $F(a)=0$, or equivalently by Proposition \ref{Prop.A_r(x)=0}, there exists
$u\in \GF{q^m}\setminus\GF{q^2}$ such that
$a=\frac{(u-u^q)^{q^2+1}}{(u-u^{q^2})^{q+1}}$.
\item\label{3} There exists $u\in \GF{Q}\setminus\GF{p^{2d}}$ such that
$a=\frac{(u-u^q)^{q^2+1}}{(u-u^{q^2})^{q+1}}$. Then the $p^d+1$
zeros in $\GF{Q}$ of $P_a(X)$ are $x_0=\frac{-1}{1+(u-u^q)^{q-1}}$
and $x_\alpha=\frac{-(u+\alpha)^{q^2-q}}{1+(u-u^q)^{q-1}}$ for
$\alpha \in \GF{p^d}$.
\end{enumerate}
 \end{lemma}

\begin{lemma}[\cite{KCM19}]\label{Am+1}
If $A_{m}(a)=0$, then  for any $x\in \GF{Q}$ such that
$x^{q+1}+x+a=0$, it holds
\[
A_{m+1}(a)=Nr_{k}^{km}(x)\in \GF{p^d}.
\]
Furthermore, for any $t\geq 0$
\begin{equation}\label{Am_recur}
A_{m+t}(a)=A_{m+1}(a)\cdot A_{t}(a).
\end{equation}
\end{lemma}
In \cite{KCM19}, it is remained as an open problem to explicitly
compute the $p^d+1$ rational zeros.

\section{Completing the case $N_a=p^d+1$}\label{complete}
Thanks to Lemma~\ref{Lem.p^d+1cond}, throughout this section we
assume $F(a)=0$ i.e. $$A_m(a)=0.$$ Let
$$L_a(X):=X^{q^2}+X^q+aX \in \GF{Q}[X].$$ Define the sequence of polynomials
$\{B_r(X)\}$ as follows:
\begin{equation}\label{eq_12}
B_1(X)=0, B_{r+1}(X)=-a\cdot A_r(X)^q.
\end{equation}
From Lemma~\ref{Am+1} and the definition \eqref{eq.defA} it follows
\begin{equation}\label{Bm}
B_m(a)=-aA_{m-1}(a)^q=A_{m+1}(a)^{\frac{1}{q}}\in\GF{p^d}.
\end{equation}
Using \eqref{Am_recur} and an induction on $l$ it is easy to check:
\begin{proposition}
\begin{equation}\label{eq_13_1}
B_{l\cdot m}(a)=B_m(a)^l.
\end{equation}
for any integer $l\geq 1$.
\end{proposition}

The first step to solve the open problem is to induce
\begin{lemma}\label{lem_5} For any integer $r\geq 2$, in the ring $\GF{Q}[X]$ it holds
\begin{equation}\label{eq_14}
X^{q^r}=\sum_{i=1}^{r-1}A_{r-i}(a)^{q^i}\cdot
L_a(X)^{q^{i-1}}+A_r(a)\cdot X^q+B_r(a)\cdot X.
\end{equation}
\end{lemma}
\begin{proof}
The equality~\eqref{eq_14} for $r=2$ is $X^{q^2}=L_a(X)-X^q-aX$
which is valid by the definition of $L_a(X)$. Suppose the
equality~\eqref{eq_14} holds for $r\geq 2$. By raising $q-$th power
to both sides of the equality~\eqref{eq_14}, we get
\begin{align*}
X^{q^{r+1}}&=\sum_{i=1}^{r-1}A_{r-i}(a)^{q^{i+1}}\cdot
L_a(X)^{q^i}+A_r(a)^q\cdot X^{q^2}+B_r(a)^q\cdot X^q\\
&=\sum_{i=2}^{r}A_{r+1-i}(a)^{q^{i}}\cdot
L_a(X)^{q^{i-1}}+A_r(a)^q\cdot X^{q^2}+B_r(a)^q\cdot X^q\\
&=\sum_{i=2}^{(r+1)-1}A_{r+1-i}(a)^{q^i}\cdot
L_a(X)^{q^{i-1}}+A_r(a)^q\cdot L_a(X)-A_r(a)^q\cdot X^q\\
&-a\cdot A_r(a)^q\cdot x+B_r(a)^q\cdot X^q\\
&=\sum_{i=1}^{(r+1)-1}A_{r+1-i}(a)^{q^i}\cdot
L_a(X)^{q^{i-1}}+A_{r+1}(a)\cdot X^q+B_{r+1}(a)\cdot X,
\end{align*}
where the last equality follows from the definitions \eqref{eq_12}
and \eqref{eq.defA}. This shows that the equality~\eqref{eq_14}
holds also for $r+1$ and so for all $r\geq 2$.\qed
\end{proof}
For $r=m$, under the assumption $A_m(a)=0$,  Lemma \ref{lem_5} gives
\[
X^{q^m}=\sum_{i=1}^{m-1}A_{m-i}(a)^{q^i}\cdot
L_a(X)^{q^{i-1}}+B_m(a)\cdot X.
\]
Now, we define
\begin{equation}\label{def_f1}
F_1(X):=X^{q^m}-B_m(a)\cdot X=\sum_{i=1}^{m-1}A_{m-i}(a)^{q^i}\cdot
L_a(X)^{q^{i-1}}\in \GF{p^d}[X]
\end{equation}
and
\begin{equation}\label{def_g1}
G_1(X)=\sum_{i=1}^{m-1}A_{m-i}(a)^{q^i}\cdot X^{q^{i-1}}.
\end{equation}
Then, evidently,
\begin{equation}\label{F1La}
F_1(X)=G_1\circ L_a(X).
\end{equation}
Furthermore, we can show
\begin{proposition}\label{lem_8}
\[
F_1(X)=L_a\circ G_1(X).
\]
\end{proposition}
\begin{proof}
When $m=3$, $A_3(a)=0$ is equivalent to $a=1$. Therefore, one has
$F_1(X)=X^{q^3}-X=(X^q-X)^{q^2}+(X^q-X)^{q}+(X^q-X)=L_a\circ
G_1(X)$.

Now, suppose $m\geq 4$. Then, by using Definition~\eqref{eq_12}
\begin{align*}
&L_a\circ
G_1(X)=\\&\sum_{i=1}^{m-1}A_{m-i}(a)^{q^{i+2}}\cdot X^{q^{i+1}}+\sum_{i=1}^{m-1}A_{m-i}(a)^{q^{i+1}}\cdot X^{q^{i}}+\sum_{i=1}^{m-1}aA_{m-i}(a)^{q^i}\cdot X^{q^{i-1}}\\
&=\sum_{i=2}^{m}A_{m+1-i}(a)^{q^{i+1}}\cdot X^{q^{i}}+\sum_{i=1}^{m-1}A_{m-i}(a)^{q^{i+1}}\cdot X^{q^{i}}+\sum_{i=0}^{m-2}aA_{m-1-i}(a)^{q^{i+1}}\cdot X^{q^{i}}\\
&=X^{q^m}-B_m(a)\cdot X=F_1(X),
\end{align*}
where Equality~\eqref{eq.defAA} was exploited to deduce the last
second equality. \qed
\end{proof}

By \eqref{Am_recur}, from $A_m(a)=0$ it follows $A_{l\cdot m}(a)=0$
for any $l\geq 1$. Therefore, \eqref{eq_13_1} and \eqref{eq_14} for
$r=lm$ yield that for any $l\geq 1$
\begin{equation}\label{eq_16}
X^{q^{l\cdot m}}-B_m(a)^l\cdot X=\sum_{i=1}^{l\cdot m-1}A_{l\cdot
m-i}(a)^{q^i}\cdot L_a(X)^{q^{i-1}}.
\end{equation}
\begin{proposition}\label{lem_6}
Relation (\ref{eq_16}) can be rewritten by using $F_1(X)$ as
follows:
\begin{equation}\label{eq_17}
X^{q^{l\cdot m}}-B_m(a)^l\cdot
X=\sum_{i=0}^{l-1}B_{m}(a)^{l-1-i}\cdot F_1(X)^{q^{m\cdot i}}.
\end{equation}
\end{proposition}
\begin{proof}
If $l=1$, the equality is equivalent to the definition of $F_1(X)$.
Suppose that it holds for $l\geq 2$. By raising $q^m-$th power to
both sides of (\ref{eq_17}), we have
\begin{align*}
X^{q^{(l+1)m}}-B_m(a)^l\cdot
X^{q^m}&=\sum_{i=0}^{l-1}B_{m}(a)^{l-1-i}\cdot
F_1(X)^{q^{m\cdot (i+1)}}\\
&=\sum_{i=1}^{(l+1)-1}B_{m}(a)^{(l+1)-1-i}\cdot F_1(X)^{q^{m\cdot
i}}.
\end{align*}
Since
\begin{align*}
X^{q^{(l+1)m}}-B_m(a)^l\cdot X^{q^m}&=X^{q^{(l+1)m}}-B_m(a)^l\cdot
F_1(X)-B_m(a)^{l+1}\cdot X,
\end{align*}
one has
\begin{align*}
X^{q^{(l+1)m}}-B_m(a)^{l+1}\cdot
X&=\sum_{i=1}^{(l+1)-1}B_{m}(a)^{(l+1)-1-i}\cdot F_1(X)^{q^{m\cdot
i}}+B_m(a)^l\cdot F_1(X)\\
&=\sum_{i=0}^{(l+1)-1}B_{m}(a)^{(l+1)-1-i}\cdot F_1(X)^{q^{m\cdot
i}}
\end{align*}
This shows that Equality~\eqref{eq_17} holds for all $l\geq 1$.\qed
\end{proof}

 Define
\begin{equation}\nonumber
\begin{array}{c}
N:=(p^d-1)\cdot m,\\
G_2(X)=\sum_{i=0}^{p^d-2}B_{m}(a)^{p^d-2-i}\cdot X^{q^{m\cdot i}}.
\end{array}
\end{equation}
Since $F_1(X)$ and $G_2(X)$ are $p^d-$linearized polynomials over
$\GF{p^d}$, they are commutative under the symbolic multiplication
``$\circ$" (see e.g. 115 page in \cite{Lidl1997}). Therefore,
regarding Equation~\eqref{eq_17} and Proposition~\ref{lem_8}, one
has
\begin{equation}\label{eq_18_0}
X^{q^N}-X=G_2\circ F_1(X)=F_1\circ G_2(X)=L_a\circ G_1\circ G_2(X)
\end{equation}
and consequently
\begin{equation}\label{eq_18}
\ker(F_1)=G_2(\GF{q^N}),
\end{equation}
\begin{equation}\label{eq_19}
\ker(L_a)=G_1\circ G_2(\GF{q^N}).
\end{equation}
Since $L_a(X)=XP_a(X^{q-1})$, here we can state:
\begin{proposition} For $a\in \GF{Q}^*$,
\begin{equation}\label{closedsol}
\{x\in \overline{\GF{p}}\mid x^{q+1}+x+a=0\}=\{x^{q-1}\mid x\in
G_1\circ G_2(\GF{q^N})\}\setminus \{0\}.
\end{equation}
\end{proposition}

Our goal is to determine $S_a=\{x\in \GF{Q} \mid P_a(x)=0\}$, the
set of all $\GF{Q}-$zeros to $P_a(X)=X^{q+1}+X+a=0, a\in\GF{Q}$.

\begin{remark}\label{lem_9} In order to find the
$\GF{Q}-$zeros of $P_a(X)$ it is not enough to consider the
$\GF{Q}-$zeros of $L_a(X)$. In fact,  one can see that $B_m(a)\neq
1$ in general. However, it holds:
\begin{proposition}
$L_a(X)=0$ has a solution in $\GF{Q}^*$ if and only if $B_m(a)=1$.
\end{proposition}
\begin{proof}
If $L_a(x)=0$ for $x\in\GF{Q}^*$, then by \eqref{F1La} $F_1(x)=0$
i.e. $x^{q^m}-B_m(a)\cdot x=(1-B_m(a))\cdot x=0$ and consequently
$B_m(a)=1$. Conversely, assume $B_m(a)=1$. Then
$F_1(X)=X^{q^m}-X=L_a\circ G_1(X)$ and $\ker(L_a)=G_1(\GF{q^m})$.
Assume $G_1(\GF{Q})=\{0\}$. Then, since $G_1$ is $q-$linearized, it
holds $G_1(\GF{q^m})=G_1([\GF{q}, \GF{Q}])=\{0\}$ which contradicts
to $\deg(G_1)< q^m$. Thus there exists such a $x_0\in \GF{Q}^*$ that
$G_1(x_0)\neq 0$. Then $G_1(x_0)\in \ker(L_a)\cap \GF{Q}^*$.
\end{proof}

\end{remark}

To achieve the goal, we will further need the following lemmas.

\begin{lemma}\label{lem_10}
Let $L(X)$ be any $q-$linearized polynomial over $\GF{Q}$. If
$x_0^{q-1}\in \GF{Q}$, then $L(x_0)^{q-1}\in\GF{Q}$.
\end{lemma}
\begin{proof}
If $x_0^{q-1}\in \GF{Q}$ i.e. $x_0^{q-1}=\lambda$ for some
$\lambda\in \GF{Q}$, then $x_0^{q}=\lambda x_0$ and subsequently
$x_0^{q^i}=\prod_{j=0}^{i-1}\lambda^{q^j} x_0$ for every $i\geq 1$.
Therefore,  when $L(X)$ is a $q-$linearized polynomial over
$\GF{Q}$, one can write $L(x_0)=\overline{\lambda} x_0$ for some
$\overline{\lambda}\in \GF{Q}$. Thus,
$L(x_0)^{q-1}=\overline{\lambda}^{q-1}\lambda\in \GF{Q}$.\qed
\end{proof}

\begin{lemma}\label{lem_12}
Let $s=\frac{(q^m-1)\cdot(p^{d}-1)}{(Q-1)\cdot (q-1)}$. If
$A_m(a)=0$ and $x_0\in \ker(F_1)$, then $x_0^{s} \in \ker(F_1)$ and
$(x_0^{s})^{q-1}\in\GF{Q}$.
\end{lemma}
\begin{proof} For $x_0=0$, the statement is trivial. Therefore, we
can assume $x_0 \neq 0$. Then, $x_0\in \ker(F_1)$ implies
\begin{equation}\label{x0}
B_m(a)=x_0^{q^m-1}=(x_0^s)^{(q-1)\cdot \frac{Q-1}{p^{d}-1}}.
\end{equation}
Since $B_m(a)\in\GF{p^{d}}$, therefore $(x_0^s)^{q-1}\in\GF{Q}$.

Now, we will show
\[
B_m(a)=B_m(a)^s.
\]
Since $P_a(X)$ has $p^d+1$ rational solutions when $A_m(a)=0$, there
exists such a non-zero $x_1$ that
\[
L_a(x_1)=0, x_1^{q-1}\in\GF{Q}.
\]
Then \eqref{F1La} gives $F_1(x_1)=0$ i.e.
\[
x_1^{q^m-1}=B_m(a),
\]
and on the other hand
\[
x_1^{q^m-1}=(N_{\GF{Q}|\GF{p^{d}}}(x_1^{q-1}))^s=(N_{\GF{q^{m}}|\GF{q}}(x_1^{q-1}))^s=(x_1^{q^m-1})^s=B_m(a)^s,
\]
where the second equality followed from the fact that
$N_{\GF{Q}|\GF{p^{d}}}(y)=N_{\GF{q^{m}}|\GF{q}}(y)$ for any $y\in
\GF{Q}$. Thus, $B_m(a)=B_m(a)^s$.

Hence, $(x_0^s)^{q^m-1}=(x_0^{q^m-1})^s=B_m(a)^s=B_m(a)$ i.e.
$F_1(x_0^s)=0.$\qed
\end{proof}

Now, take  any $x_0\in \ker(F_1)$. The definition \eqref{def_f1} and
Lemma~\ref{lem_12} shows
$$x_0^s\cdot\GF{Q}^*:=\{x_0^s\cdot\alpha \mid \alpha\in \GF{Q}^*\}\subset \ker(F_1)=G_2(\GF{p^N})$$ and
$$(x_0^s\cdot\GF{Q}^*)^{q-1}\subset \GF{Q}.$$
Subsequently, Lemma~\ref{lem_10} and Equality~\eqref{closedsol}
prove
\[
G_1(x_0^s\cdot\GF{Q}^*)^{q-1}\subset S_a.
\]
In order to avoid the trivial zero solution, we need
\[
G_1(x_0^s\cdot\GF{Q}^*)\neq\{0\}.
\]
In fact, this is the case. Really, if we assume
$G_1(x_0^s\cdot\GF{Q}^*)=\{0\}$, then $G_1(x_0^s\cdot
\GF{q^m})=\{0\}$ (because $G_1$ is $\GF{q}-$linear, and $\GF{q^m}$
is generated by $\GF{q}$ and $\GF{Q}$) which contradicts to
$\deg(G_1)< q^m$.

Next, in order to explicit all $p^d+1$ elements in $S_a$, we need to
deduce the following lemma.
\begin{lemma}
Let $A_m(a)=0$ and $x_0$ be a $\GF{Q}-$solution to $P_a(X)=0$. Then,
$\frac{x_0^2}{a}$ is a $(q-1)-$th power in $\GF{Q}$. For $\beta\in
\GF{Q}$ with $\beta^{q-1}=\frac{x_0^2}{a}$,
\begin{equation}\label{rw}
w^{q}-w+\frac{1}{\beta x_0}=0
\end{equation}
has exactly $p^d$ solutions in $\GF{Q}$. Let $w_0\in \GF{Q}$ be a
$\GF{Q}-$solution to Equation~\eqref{rw}. Then, the $p^d+1$
solutions in $\GF{Q}$ to $P_a(X)=0$ are $x_0,
(w_0+\alpha)^{q-1}\cdot x_0$ where $\alpha$ runs over $\GF{p^d}$.
\end{lemma}
\begin{proof} We substitute $x$ in $P_a(x)$ with $x_0-x$ to get
$$(x_0-x)^{q+1}+(x_0-x)+a=0$$ or
$$x^{q+1}-x_0x^q-x_0^qx-x+x_0^{q+1}+x_0+a=0$$
which implies
$$x^{q+1}-x_0x^q-(x_0^q+1)x=0,$$
or equivalently,
$$x^{q+1}-x_0x^q+\frac{a}{x_0}x=0.$$

 Since $x=0$ corresponds to $x_0$ being a zero of
$P_a(X)$, we can the latter equation by $x^{q+1}$ to get
\begin{equation}\label{y}
\frac{a}{x_0}y^q-x_0y+1=0
\end{equation} where $y=\frac{1}{x}$. Now, let
$y=tw$ where
\begin{equation}\label{t}
t^{q-1}=\frac{x_0^2}{a}.
\end{equation}
 Then,
Equation~\eqref{y} is equivalent to
\begin{equation}\label{w}
w^{q}-w+\frac{1}{tx_0}=0.
\end{equation}
If $t_0$ is a solution to Equation~\eqref{t}, then the set of all
$q-1$ solutions can be represented as $t_0\cdot \GF{q}^*$.  For
every $\lambda\in \GF{q}^*$, when $w_0$ is a solution to
Equation~\eqref{w} for $t=t_0$, $\lambda w_0$ is a solution to
Equation~\eqref{w} for $t=t_0/\lambda$. By the way, $(t_0,w_0)$ and
$(t_0/\lambda, \lambda w_0)$ give the same $y_0=t_0\cdot
w_0=t_0/\lambda \cdot \lambda w_0$. Therefore, to find all
$\GF{Q}-$solutions to Equation~\eqref{y} one can consider
Equation~\eqref{w} for any fixed solution $t_0$ of
Equation~\eqref{t}.

Now, we will show that any solution $t_0$ to Equation~\eqref{t} lies
in $\GF{q}\cdot \GF{Q}:=\{\alpha\cdot \beta\mid \alpha\in \GF{q},
\beta\in \GF{Q}\}$. In fact, we know that Equation~\eqref{w} has
$p^d$ solutions $w$ with $y=wt_0\in \GF{Q}$. Let's fix a solution
$w_0$ with $y_0=w_0t_0\in \GF{Q}$ of Equation~\eqref{w}. Then, the
set of all solutions to Equation~\eqref{w} can be written as
$w_0+\GF{q}$. Therefore, it follows that there exist $p^d\geq 2$
elements $\lambda\in \GF{q}$ with $(w_0+\lambda)t_0\in \GF{Q}$. As
$w_0t_0\in \GF{Q}$ and $(w_0+\lambda)t_0\in \GF{Q}$, we have
$\lambda t_0\in \GF{Q}$ i.e. $t_0\in \frac{1}{\lambda}\GF{Q}\subset
\GF{q}\cdot \GF{Q}$.

Hence, we can write $t_0=\alpha\cdot \beta$, where $\alpha\in
\GF{q}, \beta\in \GF{Q}$, and it follows that the set of all
solutions to Equation~\eqref{t} are $\GF{q}^*\cdot \beta$. This
means that Equation~\eqref{t} has $p^d-1$ solutions (i.e.
$\GF{p^d}^*\cdot \beta$) in $\GF{Q}$, i.e., $\frac{x_0^2}{a}$ is a
$(q-1)-$th power in $\GF{Q}$. Moreover, Equation~\eqref{rw} has
exactly $p^d$ solutions in $\GF{Q}$ (because Equation~\eqref{y} has
exactly $p^d$ solutions $y=w\beta$ in $\GF{Q}$). When $w_0\in
\GF{Q}$ is such a solution, the set of all $p^d$ solutions in
$\GF{Q}$ is $w_0+\GF{p^d}$. Since Equation~\eqref{w} yields
$y=wt=\frac{1}{(1-w^{q-1})x_0}$, we have
$x_0-x=x_0-\frac{1}{y}=x_0-(1-w^{q-1})x_0=w^{q-1}x_0$. The proof is
over. \qed
\end{proof}

Finally, all discussion of this section are summed up in the
following theorem.
\begin{theorem} Assume $A_m(a)=0$.
Let $N=m(p^d-1)$, $s=\frac{(q^m-1)\cdot(p^{d}-1)}{(Q-1)\cdot
(q-1)}$, $G_1(X)=\sum_{i=0}^{m-2}A_{m-1-i}(a)^{q^{i+1}}\cdot
X^{q^{i}}$ and $G_2(X)=\sum_{i=0}^{p^d-2}B_{m}(a)^{p^d-2-i}\cdot
X^{q^{mi}}$. It holds $
G_1(G_2(\GF{p^N}^*)^s\cdot\GF{q}^*\cdot\GF{Q}^*)^{q-1}\neq \{0\}$.
Take a $x_0\in
G_1(G_2(\GF{p^N}^*)^s\cdot\GF{q}^*\cdot\GF{Q}^*)^{q-1}\setminus
\{0\}$. $\frac{x_0^2}{a}$ is a $(q-1)-$th power in $\GF{Q}$. For
$\beta\in \GF{Q}$ with $\beta^{q-1}=\frac{x_0^2}{a}$,
\begin{equation}
w^{q}-w+\frac{1}{\beta x_0}=0
\end{equation}
has exactly $p^d$ solutions in $\GF{Q}$. Let $w_0\in \GF{Q}$ be a
$\GF{Q}-$solution to Equation~\eqref{rw}. Then, the $p^d+1$
solutions in $\GF{Q}$ of $P_a(X)$ are $x_0, (w_0+\alpha)^{q-1}\cdot
x_0$ where $\alpha$ runs over $\GF{p^d}$.
\end{theorem}

Note that one can also explicit $w_0$ by an immediate corollary of
Theorem 4 and Theorem 5 in \cite{MKCL2019}.

\section{Conclusion}
In
\cite{Bluher2004,HK2008,HK2010,BTT2014,Bluher2016,KM2019,CMPZ2019,MS2019,KCM19},
 partial results about the zeros of $P_a(X)=X^{p^k+1}+X+a$ over $\GF{p^n}$ have been obtained. In this paper, we provided explicit
expressions for all possible zeros in $\GF{p^n}$ of $P_a(X)$ in
terms of $a$ and thus finalize the study initiated in these papers.

\section*{Acknowledgement}
The authors deeply thank Professor Dok Nam Lee for his
many helpful suggestions and careful checking.

\end{document}